\documentclass{eptcs}
\usepackage{amssymb,amsmath,amsthm}
\usepackage{graphicx}

\newtheorem{proposition}{Proposition}
\newtheorem{lemma}{Lemma}
\newtheorem{theorem}{Theorem}
\newtheorem{corollary}{Corollary}
\theoremstyle{definition}

\newtheorem{example}{Example}
\newtheorem{definition}{Definition}

\newcommand{\length}[1]{\vert #1 \vert}

\newcommand{\fact}[2]{\mathcal{F}_{#1}(#2)}
\newcommand{\shift}[2]{\sigma^{\,#2}(#1)}
\newcommand{\modd}{~\mathrm{mod}~}
\newcommand{\alphabet}{\Sigma}
\newcommand{\emptyword}{\varepsilon}
\newcommand{\nonnegint}{\mathbb{N}}

\newcommand{\fromto}[2]{#1, \ldots, #2}
\newcommand{\fromtow}[2]{#1 \ldots #2}
\newcommand{\cword}[1]{#1_{\circ}}
\newcommand{\TR}{\circ}

\renewcommand{\emph}[1]{\textsl{#1}}

\title{Representations of Circular Words}
\author{L\'{a}szl\'{o} Heged\"{u}s\footnotemark[1] \qquad\qquad\qquad Benedek Nagy\footnotemark[1]~ \footnotemark[2]
\email{\{hegedus.laszlo,~nbenedek\}@inf.unideb.hu}
\institute{\footnotemark[1]~ Department of Computer Science,\\Faculty of Informatics, University of Debrecen}
\institute{\footnotemark[2]~ Department of Mathematics, Faculty of Arts and Sciences,\\
Eastern Mediterranean University, Famagusta, North Cyprus, Mersin-10, Turkey}}

\def\titlerunning{Representations of Circular Words}
\def\authorrunning{L. Heged\"{u}s ~\&~ B. Nagy}

\begin{document}

\maketitle

\begin{abstract}
    In this article we give two different ways of representations of
    circular words. Representations with tuples are intended as a
    compact notation, while representations with trees give a way to
    easily process all conjugates of a word. The latter form can also
    be used as a graphical representation of periodic properties
    of finite (in some cases, infinite) words. We also define iterative
    representations which can be seen as an encoding utilizing the flexible
    properties of circular words. Every word over the two letter alphabet
    can be constructed starting from $ab$ by applying the fractional
    power and the cyclic shift operators one after the other, iteratively.
\end{abstract}

\section{Introduction}

 One of the most popular areas of research in theoretical computer science
 is combinatorics on words. This field deals with various properties of
 finite and infinite sequences or words. Being closely related to mathematics,
 it has connections to algebra, number theory, game theory and several
 others. Although it was written decades ago, the books of M. Lothaire
 are good reads and are recommended for researchers who want to get a
 deep overview of the subject \cite{Lothaire83,Lothaire02,Lothaire05}.
 Axel Thue contributed the first results to the field \cite{Thue06,Thue12}.
 Since then many applications in computer science have been discovered
 (e.g., in string matching, data compression, bioinformatics, etc.).

 We deal with circular words (sometimes called necklaces \cite{Smyth03}
 or cyclic words) that are different from linear ones and lead to some
 interesting new viewpoints. Similar sequences can appear in nature, for
 example, the DNA sequences of some bacteria has a similar form to a necklace.
 In the simplest sense, circular words are strongly periodic discrete
 functions.
 
 Circular words are not as widely investigated as linear words. We hope
 that our approach and results may show that interesting facts can be
 obtained by analyzing these sequences.
 Dirk Nowotka wrote about unbordered conjugates of words in Chapter 4
 of his dissertation \cite{Nowotka04}. Complementing this, we deal with
 bordered conjugates that have periods smaller than the length of the
 word.
 Another related article is \cite{FazNagy08}, where permutations and
 cyclic permutations of primitive and non-primitive words were
 investigated.
 For an overview of current research about circular words, the reader can consult
 the following articles. Relations to Weinbaum factorizations are investigated
 in \cite{DHN06}. Several articles were written about pattern avoidance
 of circular words, for example, \cite{Curie02,Fitzpatrick04,Shur10} to
 name a few.
 Other applications in mathematics, namely integer sequences \cite{RittaudVivier11,RittaudVivier12}
 were also considered.
 
 The notion of weak and strong periods was introduced in  \cite{Hegedus13}.
 One result about periodic functions is often cited in combinatorics on
 words, since it is clearly about periodic infinite words too. This result
 belongs to Fine and Wilf \cite{FineWilf65}. It can be shown by example
 that this statement is not true for weak periods of circular words \cite{Hegedus13}.
 In this paper, we investigate two kinds of representations of circular
 words continuing the research line of the paper \cite{Hegedus13}
 presented at the WORDS 2013 conference in Turku. The first one is connected
 to the property that every linear word has a shortest root, while the other
 one is related to tries (see e.g., \cite{Smyth03}).
 
 The structure of the paper is as follows.
 Section \ref{Prelim} defines the notation and notions used in the rest
 of the article. After this, in Section \ref{tuples_reps} we discuss
 ways of representing circular words with tuples and an algorithm to
 construct one of these representations. Section \ref{trees_reps} is about
 representing circular words with trees (or tries) and we present some
 results related to Fibonacci words.
 At the end in Section \ref{concs} some possible directions of future
 research is discussed.

\section{Preliminaries}\label{Prelim}

The following notions and notation are used in the rest of the article.
We will call a non-empty set of symbols an \emph{alphabet} and denote it
by $\alphabet$. \emph{Words} (or \emph{linear words}) over $\alphabet$
are finite sequences of symbols of $\alphabet$. The operation of concatenation
is defined by writing two words after each-other. The \emph{empty word},
i.e., the empty sequence is denoted by $\emptyword$ and it is the unit
element of the monoid $\alphabet^{*}$. We also define
$\alphabet^+ = \alphabet^* \setminus \{\emptyword\}$.
The \emph{length} of the word
$w \in \alphabet^*$ (denoted by $\length{w}$) is the length of $w$ as a
sequence, that is, the number of all the symbols in $w$.
We will use $\nonnegint$ to denote the set of non-negative integers.

We say, that $v\in \alphabet^{*}$ is a \emph{factor} of $w \in \alphabet^{*}$
if there exist words $x,y \in \alphabet^{*}$ such that $w = xvy$.
Furthermore, if $x=\emptyword$ (resp. $y=\emptyword$), then $v$ is a \emph{prefix}
(resp. \emph{suffix}) of $w$. For any word $w$ and integer
$0 \leq k \leq \length{w}$, we denote the \emph{length $k$ factors} of
$w$ by $\fact{k}{w}$.
For arbitrary positive integers $p$ and $q$, we use $(p \modd q)$ to
denote the remainder of $\frac{p}{q}$.
Let $w \in \alphabet^{*}$ be a word of length $n$, that is,
$w = \fromtow{w_{1}}{w_{n}}$, where $\fromto{w_{1}}{w_{n}} \in \alphabet$.
Then for any $p \in \nonnegint$, we have
$w^{\frac{p}{n}} = w^{\lfloor\frac{p}{n}\rfloor}w'$, where
$w' = w_1\ldots w_{(p \modd n)}$. We call $w^{\frac{p}{n}}$ the
\emph{fractional power} of $w$.
From now on we will always refer to
the $i$th position of a word $w \in \alphabet^*$ as $w_i$.
A word $w\in \alphabet^+$ is \emph{primitive} if there is no word
$v\in \alphabet^*$ such that $w = v^p$ where $p\in \nonnegint$, $p > 1$.

A positive integer $p$ is a \emph{period} of $w=w_1\ldots w_n$ if
$w_i = w_{i+p}$ for all $i = 1, \ldots, n-p$. As a complementary notion,
word $v\in \alphabet^*$ is a \emph{border} of $w \in \alphabet^*$ if $v$
is a prefix and also a suffix of $w$. Each word $w \in \alphabet^{*}$
has \emph{trivial borders} $\emptyword$ and $w$.
It is clear, that word $w$ has a
border $b$ if and only if $w$ has period $\length{w} - \length{b}$.

Words $x$ and $y$ are \emph{conjugates} if there exist
words $u,v\in\alphabet^*$ such that $x=uv$ and $y=vu$.
Related to this notion, we define the \emph{shift} operation
$\shift{w}{ }$
for all $w \in \alphabet^{*}$ as follows:
\[
\shift{w}{ } = \fromtow{w_{2}}{w_{n}}w_{1}.
\]
Moreover, $\shift{w}{\ell} = \shift{\shift{w}{ }}{\ell-1} = 
\fromtow{w_{1 + \ell}}{w_{n}}\fromtow{w_{1}}{w_{\ell}}$.
Also, we will use $\shift{w}{-\ell}$ that can also be written as $\shift{w}{\length{w}-\ell}$.

Lyndon and Sch\"{u}tzenberger stated the following, which characterizes
the relation between a word and its non-trivial borders
\cite{LyndonSchutzenberger62}.

\begin{lemma}[Lyndon and Sch\"{u}tzenberger]\label{LyndonSchutzenberger}
Let $x \in \alphabet^+$, $y$, $b \in \alphabet^*$ be arbitrary words.
Then $xb = by$ if and only if there exist
$u \in \alphabet^+$, $v \in \alphabet^*$ and $k \in \nonnegint$ such that
$x = uv$, $y = vu$ and $b = (uv)^ku = u(vu)^k$.
\end{lemma}

A \emph{circular word} is obtained from a linear word
$w \in \alphabet^*$ if we link its first symbol after the last one, as
seen on Figure \ref{create_circular}.

\begin{figure}
\begin{center}
\includegraphics[scale=1.2]{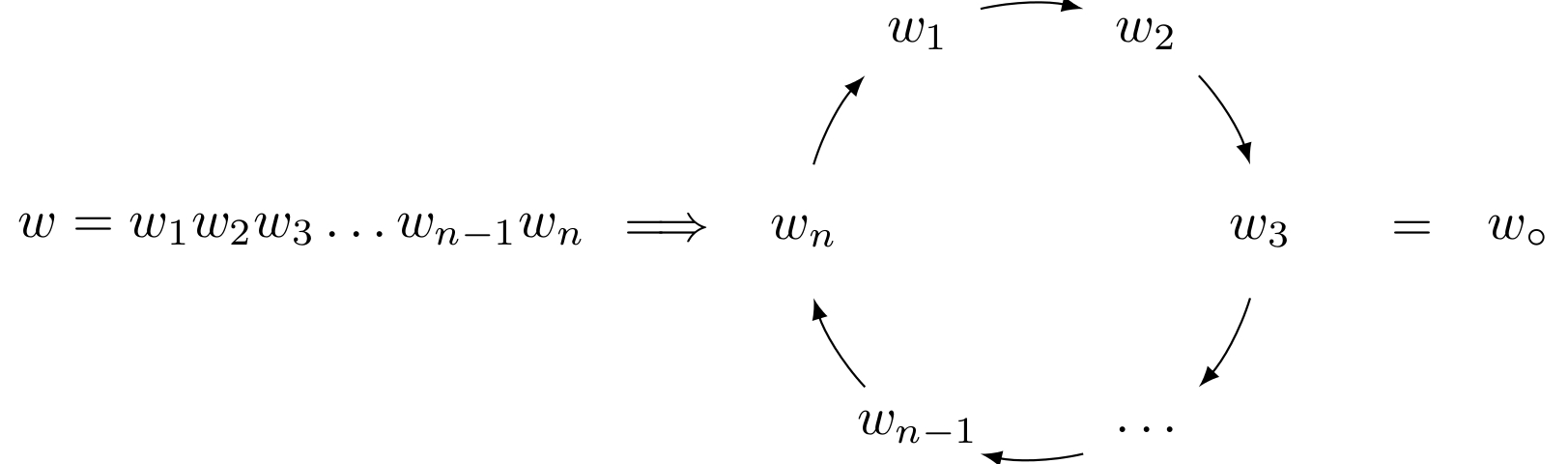}
\caption{Creating the circular word $\cword{w}$ from the linear word $w$.\label{create_circular}}
\end{center}
\end{figure}

One can see from the figure that circular words do not have a beginning
nor an end. Nor do the notions of suffix and prefix make sense.
A circular word $\cword{w}$ can be seen as the set of all conjugates of $w$, 
or all cyclic shifts of $w$, that is, the set
\[
\cword{w} = \{ v ~|~ v~\mbox{ is a conjugate of }~w\} =
\{\shift{w}{\ell}~|~ \ell = 0, \ldots, \length{w}-1\}.
\]

Note, that $\cword{w}$ consists exactly of the length $\length{w}$ factors of
$ww$. That is, $\cword{w} = \fact{\length{w}}{ww}$.
The notions of weak- and strong periods were given in \cite{Hegedus13}.
We will only refer to weak periods in this paper and define them as
follows.

\begin{definition}
The positive integer $p$ is a \emph{weak} (\emph{strong})
\emph{period} of a circular word
$\cword{w}$ if $p$ is a period of at least one (all) of the
conjugates $v \in \cword{w}$.
\end{definition}

\section{Representations with tuples}\label{tuples_reps}

If not stated otherwise, we assume that alphabet $\alphabet$ can be
arbitrary.
Every word $w \in \alphabet^{*}$ can be represented by a power of a
(possibly shorter) word $u \in \alphabet^{*}$ and a positive integer that is the
length of $w$. In other words, for all $w \in \alphabet^{*}$, there
exists a word $u \in \alphabet^{*}$ such that
$u^{\frac{\length{w}}{\length{u}}} = w$. We will call such a $u$ a \emph{root} of $w$, while the shortest root
is called the \emph{primitive root} of $w$ (see e.g., pages 10--11 of \cite{Smyth03}).
In this section we discuss analogous representations of circular words that
take advantage of their lack of strictly specified endpoints.

\begin{definition}
A pair $(u,n) \in \alphabet^* \times \nonnegint$ is a
\emph{representation} of the circular word $\cword{w}$ over $\alphabet$ if
$\length{u} \leq n$, $n = \length{\cword{w}}$ and
$u^{\frac{n}{\length{u}}} \in \cword{w}$.
\end{definition}

\begin{definition}
A \emph{minimal representation} of a circular word $\cword{w}$ over
$\alphabet$ is a representation $(u,n)$ of $\cword{w}$, such that
$\length{u} \leq \length{u'}$ for any other representation $(u',n)$ of
$\cword{w}$.
\end{definition}

It is clear, that every circular word has a minimal representation,
since all of them have a smallest weak period.
Trivially, that not all pairs $(u,n)$ are minimal
representations of some circular word. For example, consider the
representation $(baa,5)$ of the circular word $\cword{(baaba)}$. This
circular word also has a representation $(ab,5)$ which is in fact a
minimal representation.

It is also true, that a circular word can have more than one minimal
representations. For example, $(ababa,12)$, $(babaa,12)$, $(abaab,12)$
and $(baaba,12)$ are all minimal representations of the circular word
$\cword{(ababaababaab)}$. Note, that $(aabab,12)$ is not a minimal
representation of this circular word, since it represents
$\cword{(aababaababaa)}$.

Clearly, if $n = k \cdot \length{u}$ for some $k \in \nonnegint$ in a
minimal representation $(u,n)$, then $(\shift{u}{\ell},n)$ is also a
minimal representation of the same circular word for all
$\ell = 0, \ldots, \length{u} - 1$.

Suppose, that $w = u^mu'$ for some $u \in \alphabet^*$ where $u'$ is a
non empty prefix of $u$ and $m\in \nonnegint\setminus\{0\}$. 
Then for every $k \in \nonnegint$, the word $w' = wu^k$ has a cyclic shift
$\shift{w'}{\length{w}} = u^{k+m}u'$. Thus the circular word
$\cword{w'}$ has a representation $(u,\length{w}+k\cdot\length{u})$.

\begin{theorem}
Let $(u,n)$ be a representation of $\cword{w}$. Suppose, that $u$ has border
$s$, that is, $u = sx = ys$, and $n = 2\cdot \length{u} - \length{s}$.
Then $(y,n)$ is also a representation of $\cword{w}$. Moreover, if $s$ is the
longest non-trivial border of $u$, then $(y,n)$ is a minimal
representation of $\cword{w}$.
\end{theorem}

\begin{proof}
 Let us have a representation $(u,n)$ of $\cword{w}$ that satisfies the
 assumption, that is, $u$ has border $s$ and
 $n = 2\cdot \length{u} - \length{s}$. Then $u$ is in the form
 $u = sx = ys$ for some $x,y \in \alphabet^*$ and
 $\cword{w} = \cword{(uy)} = \cword{(ysy)}$. By Lemma \ref{LyndonSchutzenberger}, $yys$
 has period $\length{y}$, thus $\cword{w}$ has weak period $\length{y}$ and a
 representation $(y,n)$.
 
 If $s$ is the longest non-trivial border of $u$, then $y$ is the
 primitive root of $u$, thus $(y,n)$ is a minimal representation of
 $\cword{w}$.
\end{proof}

Suppose that we have a representation $\cword{w} = (u,n)$, where $u \in \alphabet^{*}$
and $n \in \nonnegint$. If $\length{u} \geq 2$, then $u$ may be compressed
further. In other words, we can take a minimal representation $(u',\length{u})$
with an additional parameter $k \in \nonnegint$, such that $\shift{u}{k}$ has
primitive root $u'$. This method of compression can be done finitely many times, until
reaching a word $u_{0}$ which we will refer to as a \emph{minimal root} of $\cword{w}$.
We will call these representations \emph{iterative representations}, defined
formally in Definition \ref{iterative_reps_def}. Of course, if a minimal root of a word $\cword{w}$ has
only one letter, then it is in the form $\cword{(a^{\length{w}})}$ for some $a \in \alphabet$.
In this case, this letter is unique and we can refer to it as the minimal root of $\cword{w}$.
Thus words in these forms have trivial representations and we will no longer deal with them.

\begin{definition}\label{iterative_reps_def}
 Let $u \in \alphabet^{*}$, $m \in \nonnegint\setminus\{0\}$ and $\fromto{\ell_{1}, \ell_{2}}{\ell_{m-1}, \ell_{m}},\fromto{k_{1}, k_{2}}{k_{m-1}} \in \nonnegint$.
 The $2m$-tuple
 \[
 (u,\ell_{1},k_{1},\ell_{2},k_{2},\ldots, \ell_{m-1},k_{m-1},\ell_{m})
 \]
 is an \emph{iterative representation} of the circular word
 $\cword{w} = \cword{(u_{m-1}^{\frac{\ell_{m}}{\ell_{m-1}}})}$ over the two letter
 alphabet $\{a,b\}$, where $u_{0} = u$, $u_{1} = \shift{u_{0}^{\frac{\ell_{1}}{\length{u_{0}}}}}{k_{1}}$
 and
 $u_{i} = \shift{u_{i-1}^{\frac{\ell_{i}}{\ell_{i-1}}}}{k_{i}}$ for all $i = \fromto{2}{m-1}$.
\end{definition}

\begin{example}\label{iterat_repex}
 Consider the circular word $\cword{w} = \cword{(bababaabbabaab)}$. One of its iterative
 representations is\vspace*{-.4em}
 \[
 (baa, 4, 0, 6, 4, 14).\vspace*{-.4em}
 \]
 By using the previous definition of the words $u_{i}$, the following words are
 obtained during the reconstruction of the circular word:
 $u_{0} = baa$, $u_{1} = baab$, $u_{2} = babaab$ and finally,
 $\cword{w} = \cword{(babaabbabaabba)}$. Note, that no shifting is
 required in the last step, because $\cword{w} = \cword{v}$ for all $v \in \cword{w}$.
\end{example}

Of course, every circular word has an iterative representation of the form
above that can be constructed with the greedy algorithm in Figure \ref{iter_alg}. Moreover,
the algorithm halts if only if it has found a minimal root.

\begin{figure}[h]
\noindent\textbf{\texttt construct\_iterative\_representation($\cword{w}$)}
\vspace{-.5em}
\begin{enumerate}
 \item $u \leftarrow w$
 \item \texttt{$v \leftarrow $ find $v$ such that $(v,\length{w})$ is a minimal representation of $\cword{w}$}
 \item \texttt{$rep \leftarrow [\length{w}]$} \hfill \# rep is a vector of integers
 \item \texttt{while ~true~ do}
 \item \quad $u \leftarrow v$
 \item \texttt{\quad $v \leftarrow $ find $v$ such that $(v,\length{u})$ is a minimal representation of $\cword{u}$}
 \item \quad \texttt{if $\length{u} = \length{v}$ then} \hfill \# if we have found a minimal root,
 \item \qquad \texttt{break} \hfill \# then the algorithm breaks the loop
 \item \quad \texttt{endif}
 \item \quad \texttt{$k \leftarrow $ find $k$ such that $\shift{u}{-k}$ has root $v$}
 \item \quad \texttt{$rep \leftarrow \length{u}:k:rep$} \hfill \# append $\length{u}$ and $k$ to $rep$ from the left
 \item \texttt{endwhile}
 \item \texttt{return $v:rep$}
\end{enumerate}
\caption{Algorithm for constructing the iterative representation of $\cword{w}$.\label{iter_alg}}
\end{figure}

Note, that by using this algorithm, we can process the iterative representation in
Example \ref{iterat_repex} further to obtain $(ab, 3, 1, 4, 0, 6, 4, 14)$.
In fact, the following can be stated about the iterative representations of circular words
over the two letter alphabet $\{a,b\}$.

\begin{theorem}\label{twoletteriter}
 Let $w \in \{a,b\}^{*}$. If $(u,\ell_{1},k_{1},\ldots, \ell_{m-1},k_{m-1},\length{w})$
 is a minimal iterative representation of $\cword{w}$, then $\length{u} \leq 2$.
\end{theorem}

\begin{proof}
 It follows from the fact that every word $u \in \{a,b\}$, $\length{u} \geq 3$
 has a conjugate that has a border of length at least one, thus in this case
 $\cword{u}$ has a representation $(v,\length{u})$ such that $\length{v} < \length{u}$.
\end{proof}

Let $(u,\ell_{1},k_{1},\ldots, \ell_{m-1},k_{m-1},\length{w})$ be an iterative
representation of $\cword{w}$. It is \emph{optimal} if for all iterative representations
$(u',\ell'_{1},k'_{1},\ldots, \ell'_{m'-1},k'_{m'-1},\length{w})$ of $\cword{w}$,
$\length{u} \leq \length{u'}$ and if $\length{u} = \length{u'}$, then $m \leq m'$.
In other words an optimal iterative representation of $\cword{w}$ is one with
the shortest possible minimal root,
such that $\cword{w}$ can be reconstructed from it with the least amount of
fractional power operations (regardless of the amount of shift operations required).

The algorithm may not provide an optimal solution for all inputs $\cword{w}$.
For example, consider the circular word $\cword{(ababaa)}$. The algorithm
would construct the iterative representation $(ab, 3, 0, 4, 0, 6)$, while
an optimal solution would be $(ab, 5, 0, 6)$.
One of the directions of future research is to look for an efficient
algorithm that always finds an optimal iterative representation of
any circular word $\cword{w}$ (see Section \ref{concs}).

Note, that we do not have to restrict ourselves to representations of circular words.
If we are looking for a linear word, another shift operation has to be applied
at the end of the reconstruction.

Let us now turn to another method of representation, which is not intended
as an encoding, nor as a compression, but a way of representing the structure of different
conjugates of a word and their relation to each-other (e.g., common prefixes).

\section{Representations with trees}\label{trees_reps}

The tree $\tau$ is the tree of the circular word $\cword{w}$ if and only
if for any word $v = \fromtow{v_1}{v_n}$ in $\cword{w}$, there exists
a path in $\tau$ between the root and a leaf node
with a series of edges labeled $\fromto{v_1}{v_n}$.

This approach is related to tries that are data structures
representing associative structures. They are often used to search for
suffixes or other factors of words. Quite similarly, our trees represent
a set of words that are conjugates of each-other. For more information on
the use of tries consult \cite{Crochemore02}.

We remark, that in our figures the letters appear as nodes, but they
are to be considered as labels of edges between two (unnamed) nodes.
This way, the represented words can be seen more clearly.
First, consider the circular word
\[
\cword{(abaab)} = \{ abaab, baaba, aabab, ababa, babaa\}.
\]
Its tree representation is shown in Figure \ref{det-tree1}.

\begin{figure}
\begin{center}
\begin{minipage}[c]{.45\linewidth}
\vspace*{6em}
\begin{center}
\includegraphics[scale=1.2]{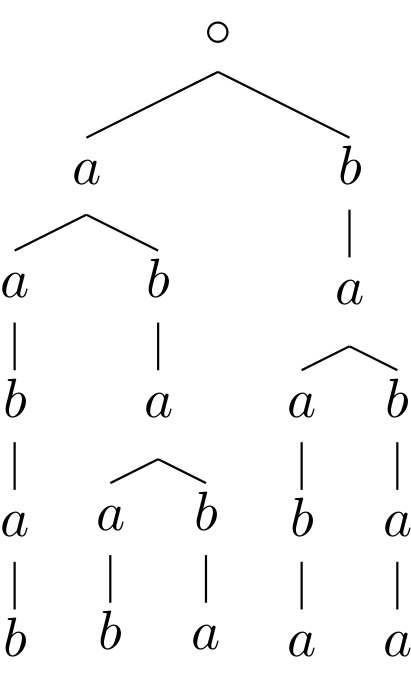}
\end{center}
\vspace*{-1em}
\caption{Tree representation of $\cword{(abaab)}$.\label{det-tree1}}
\end{minipage}
\begin{minipage}[c]{.45\linewidth}
\begin{center}
\includegraphics[scale=1.2]{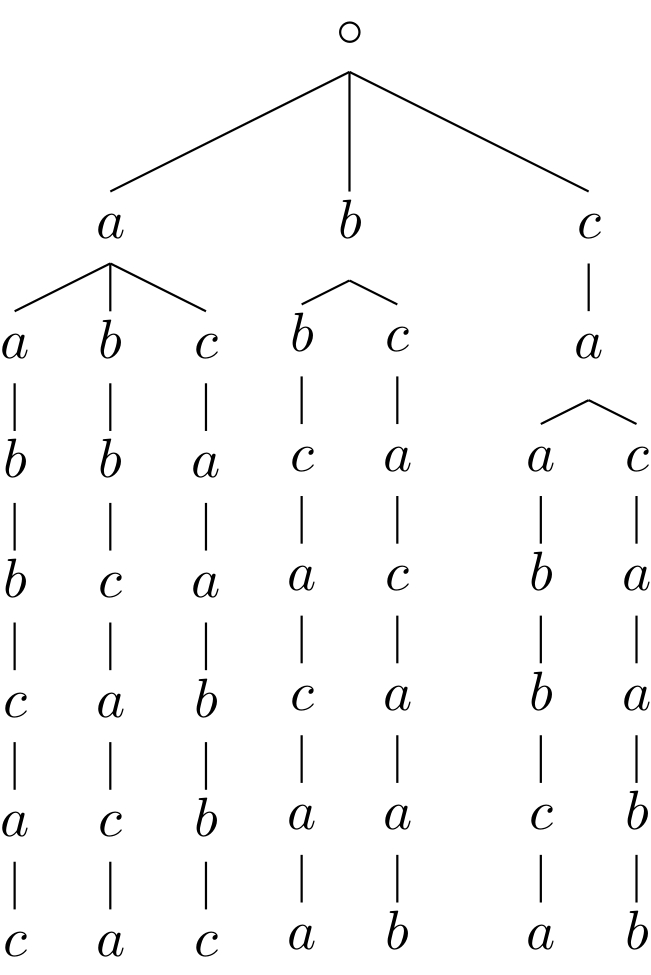}
\end{center}
\vspace*{-1em}
\caption{Tree representation of $\cword{(aabbcac)}$.\label{det-tree2}}
\end{minipage}
\end{center}
\end{figure}

Now, see Figure \ref{det-tree2} for the tree of the circular word $\cword{(aabbcac)}$ (over
the three letter alphabet $\{a,b,c\}$) which is the set
\[
\cword{(aabbcac)} = \{ aabbcac, abbcaca, bbcacaa, bcacaab, cacaabb, acaabbc, caabbca\}.
\]

Clearly, both trees represent finite-state automata with partially
defined, deterministic transition functions. We can distinguish different
levels of a tree. Vertex $\TR$ is on level zero ($\ell(\TR) = 0$) and if
there is an edge $u \to v$, then $\ell(v) = \ell(u) + 1$.

We can see some branching nodes in both trees. The tree in Figure \ref{det-tree2}
has two branching nodes on level one while no two branching nodes
of the tree in Figure \ref{det-tree1} are on the same level.

Examining branching nodes is useful for analyzing trees of
circular words and the words themselves. Suppose that tree $\tau$ has
$\fromto{u_1}{u_k}$ branching nodes such that $\TR \to^a u_1$ and
$u_i \to^a u_{i+1}$ for all $i = \fromto{1}{k-1}$. Then there is a letter $b$ such that $a^k$,
$a^{k-1}b$, and thus $\fromto{a^{k-2}b}{ab}$, $b$ are all factors of $\cword{w}$. If
the level of the leaf nodes is $k+1$, then the represented circular word
must be $\cword{(a^kb)}$. Similarly, if there are branching nodes
$\fromto{u_1}{u_m}$ and $\fromto{v_1}{v_k}$ such that
$\TR \to^a u_1 \to^a \ldots \to^a u_m$ and
$\TR \to^b v_1 \to^b \ldots \to^b v_k$,
and the level of the leaf nodes is $m+k$, then the tree can only
represent the circular word $\cword{(a^mb^k)}$. Apart from these simple cases,
we can state the following about the relation of circular words
and branching nodes in their trees:
 Let $\cword{w}$ be a circular word with tree $\tau$. There is a branching
 node in $\tau$ on level $\ell$ if and only if there are two distinct words
 $w',w'' \in \cword{w}$, such that the longest common prefix of $w'$ and $w''$
 is a word of length $\ell$.
 Moreover, if there is a branching node in the tree on level $n > 0$,
 then there is a branching node on level $n-1$. These nodes do not
 necessarily lie on the same path.
 To verify this, assume that tree $\tau$ contains the edges $u \to^a v$ and $u \to^b s$,
 where $u \neq \TR$.
 Then there are words $xay, xbz \in \cword{w}$ such that $x,y,z \in \alphabet^*$
 with $\length{x} > 0$, and
 $a,b \in \Sigma$, where $\Sigma$ is an alphabet of at least two letters.
 Write $x = \fromto{x_1}{x_m}$.
 Clearly, both $\fromtow{x_2}{x_m}ayx_1$ and
 $\fromtow{x_2}{x_m}bzx_1$ are in $\cword{w}$, having a common prefix of
 length $\length{x} - 1$. Thus there must be a node $u'$ such that the
 path from $\TR$ to $u'$ reads $\fromtow{x_2}{x_m}$ and two
 nodes $v'$ and $s'$, such that $u' \to^a v'$ and $u' \to^b s'$.

\begin{proposition}\label{branchingnodes}
 Consider a circular word $\cword{w} \in \{a, b\}$ with tree $\tau$.
 If $\tau$ has a branching node on level $\length{w} - 2$, then
 there is exactly one branching node on all levels $m = \fromto{0}{\length{w}-2}$
 of $\tau$.
\end{proposition}
\begin{proof}
 From the previous argument, it follows that all levels $k < \length{w} - 2$
 of the tree has at least one branching node. Clearly, the depth of the tree is
 $\length{w}$. Since the root node is branching, the number of possible
 paths (words) up to level one is two.
 Moreover, if level $k > 0$ has $m_{k} \in \nonnegint$ branching nodes,
 then the number of all possible paths up to level $k + 1$ is
 equal to the number of all possible paths up to level $k$, plus $m_{k}$.
 Then we get that the number of possible paths on the level of the leaf nodes is
 $2 + m_{1} + \ldots + m_{\length{w}-1} + m_{\length{w}} = \length{w}$.
 We have stated, $m_{i} > 0$ for all $i = \fromto{1}{\length{w}-2}$, thus
 $m_{\length{w}-1} = m_{\length{w}} = 0$ and
 $2 + m_{1} + \ldots + m_{\length{w}-2} = \length{w}$.
 If $m_{i} > 1$ for any $i \geq 1$, then $m_{j} = 0$
 for some $j \neq i$. This is impossible, since all levels under $\length{w}-2$
 have at least one branching node, thus $m_{i} = 1$ for all
 $i = \fromto{1}{\length{w}-2}$.
\end{proof}

Now, let us analyze an interesting class of words.
Let $f_{1} = b$, $f_{2} = a$ and define $f_{n} = f_{n-1}f_{n-2}$ for all
$n \geq 3$. We call $f_{n}$ (where $n \geq 1$) the
\emph{$n$th finite Fibonacci word}. The \emph{infinite Fibonacci word}
is the limit of the sequence $f_{1}, f_{2}, \ldots$

The following lemma describes a well known property of the infinite
Fibonacci words. 

\begin{lemma}[see S\'{e}\'{e}bold \cite{seebold85}]\label{seebold}
If a word $u^{2}$ is a factor of the infinite Fibonacci word, then u is
a conjugate of some finite Fibonacci word. \qed
\end{lemma}

Note that the tree in Figure \ref{det-tree1} represents the circular word
obtained from $f_{5}$ which is the fifth Fibonacci word. See the
trees of $\cword{(f_{6})}$ and $\cword{(f_{7})}$ in Figure \ref{treef6f7}.
One can observe that the structure of these trees are very similar.
This is strongly related to the definition of Fibonacci words.

\begin{figure}[h]
\begin{center}
\begin{minipage}{.43\linewidth}
\vspace*{2.5em}
\begin{center}
 \includegraphics[scale=1.1]{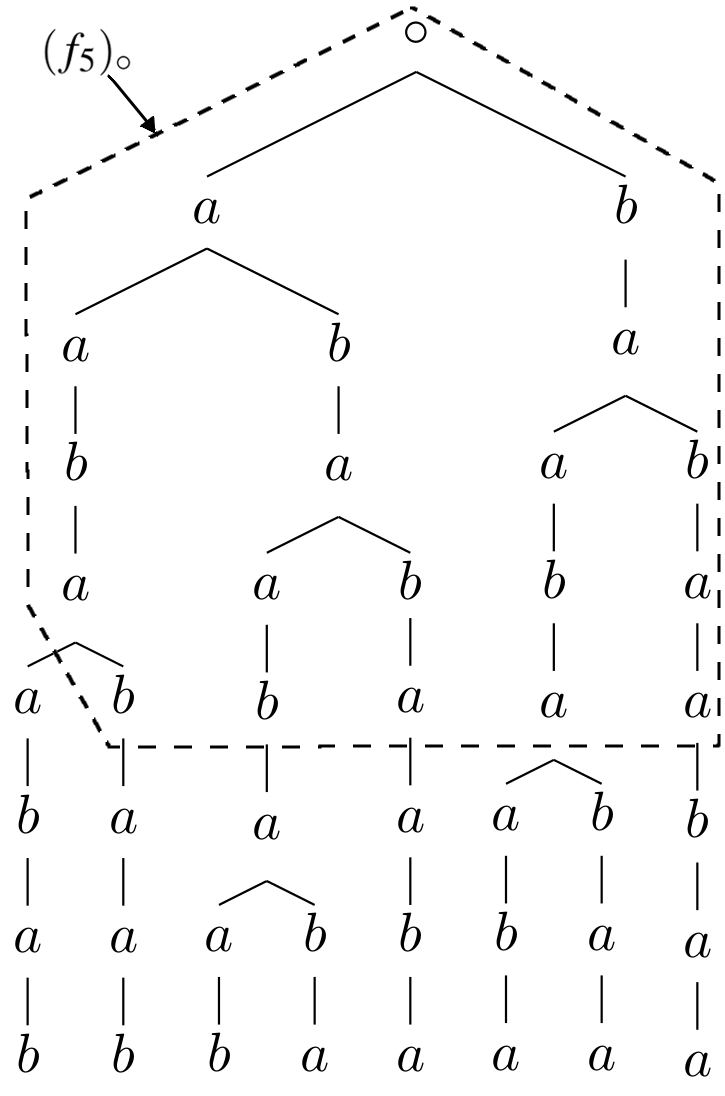}
 \end{center}
\end{minipage}
\begin{minipage}{.43\linewidth}
\begin{center}
 \includegraphics[scale=.75]{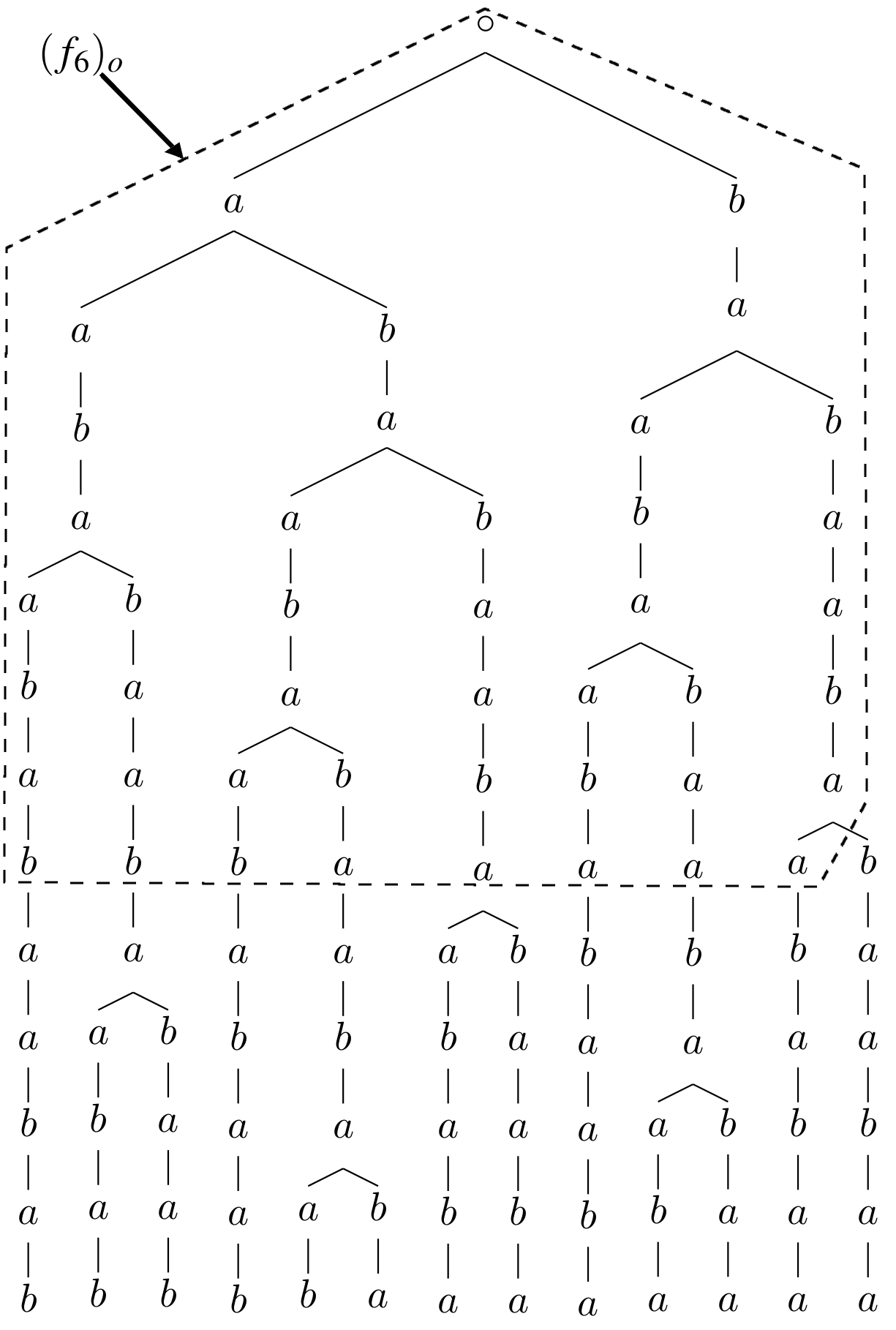}
 \end{center}
\end{minipage}
\end{center}
\vspace*{-1em}
 \caption{Trees of $(f_{6})_{o}$ and $(f_{7})_{o}$.\label{treef6f7}}
\end{figure}

\begin{theorem}\label{fibo-branching}
Let us denote the tree of the
finite Fibonacci word $f_{i}$ by $\varphi_{i}$ for all $i \in \nonnegint$.
Then for all $i \in \nonnegint$, the tree $\varphi_i$ has exactly one branching node on all
of its levels, except for the last two.
\end{theorem}
\begin{proof}
Consider the tree $\varphi_{i}$ of the circular Fibonacci  word
$\cword{(f_{i})}$ and let $\ell \in \{\fromto{0}{\length{f_{i}}}\}$.
The paths from $\TR$ to nodes on level $k$ represent the length $k$
factors of $\cword{(f_{i})}$. By the properties of Fibonacci words (or
Sturmian words), we know that the number of distinct factors of length
$k$ in the infinite Fibonacci word is $k + 1$. Since all of the
length $k$ words of the tree appear in the infinite Fibonacci word
(because it has factor $f_{i}^{2}$),
their number must not be more than $k + 1$. On the other hand, each
tree of a primitive word of length $n$ must contain $n$ branching nodes.
Thus in $\varphi_{i}$ all branching nodes must be on different levels.
\end{proof}

Based on the proof, we can state the following about the trees of circular Fibonacci words.

\begin{corollary}
For all $i,j \in \nonnegint\setminus\{0\}$, if $j > i$, then $\varphi_i$
is a subtree of $\varphi_j$.
\end{corollary}

Thus the trees of Fibonacci words are not only very similar, but they
contain recurring subtrees.
Notice in Figure \ref{treef6f7}, that the tree of $\cword{(f_5)}$ appears in the tree of $\cword{(f_6)}$ which also appears in the tree of $\cword{(f_7)}$,
marked by the dashed lines.
Thus we can define the tree $\varphi$ which belongs to the limit
of the sequence of Fibonacci words, that is, the infinite Fibonacci
word. Each path in the tree $\varphi$ defines an infinite suffix
of the infinite Fibonacci word. This is a consequence of the structure of the
trees $\varphi_{i}$ ($i = 1, 2, \ldots$), since all of their words are
factors of the infinite Fibonacci word and an infinite factor must be
a suffix.

Let us state another interesting fact about branching nodes of trees of
circular Fibonacci words.
\begin{theorem}
 Consider the tree $\varphi_{i}$ for any $i \in \nonnegint$. Let $u$ and
 $u'$ be branching nodes of $\varphi_{i}$ such that they lie on the same
 path and there are no other branching nodes between them. Then
 $|\ell(u) - \ell(u')|$ is a Fibonacci number.
\end{theorem}
\begin{proof}
 Assume the contrary, that is, there is a Fibonacci word $f_{i}$ such that
 there are two branching nodes $u$, $u'$ in tree $\varphi_{i}$ that lie on the
 same path and do not have any other branching nodes between them, but
 $|\ell(u) - \ell(u')|$ is not a Fibonacci number. Then, there exists a
 Fibonacci word $f_{j}$ with $j \geq i$ such that $\cword{(f_{j})}$ has
 square factor $vv$ where $v$ is the word constructed from the labels on
 the path between $u$ and $u'$. Moreover, this will be true for all
 Fibonacci words $f_{j'}$ where $j' \geq j$. Thus the infinite Fibonacci
 word must contain the square factor $vv$. This contradicts Lemma \ref{seebold},
 since $v$ cannot be a conjugate of any Fibonacci word because its length
 is not a Fibonacci number.
 Thus our indirect assumption is false.
\end{proof}

\section{Conclusion and future directions}\label{concs}

Combinatorics on circular words is a field that still has countless open
problems and many possible research directions. We have shown some
non-traditional methods of considering (representing) circular words.
The following questions are still open and may lead to a better
characterization of these sequences.

\begin{enumerate}
 \item The algorithm presented in Section \ref{tuples_reps} does not always provide
 optimal solutions. Is there a way of deciding how to choose the best sequence of roots in the algorithm?
 \item Theorem \ref{twoletteriter} is about the minimal roots of words over the two letter
 alphabet. What can we say about words over alphabets of more than two letters?
 \item One could use the tree $\varphi$ to deduce some properties of the
 infinite Fibonacci word.
 \item Or the tree representations can be utilized to prove results about
 the structure of other (possibly infinite) words.
 \item We believe, that Theorem \ref{fibo-branching} is true for all
 standard sturmian words (see e.g., \cite{deLucaMignosi94} for their
 definition).
\end{enumerate}

\section*{Acknowledgements}
The authors would like to thank the reviewers for their valuable and
useful comments.
The work is supported by the T\'{A}MOP 4.2.2/C-11/1/KONV-2012-0001
and 4.2.2/B-10/1-2010-0024 projects. The projects are implemented through the New
Hungary Development Plan, co-financed by the European Social Fund
and the European Regional Development Fund.

\bibliographystyle{eptcs}
\bibliography{representations}

\end{document}